\documentclass[11pt,letterpaper]{article}

\usepackage[T1]{fontenc}

\usepackage{url}
\usepackage{hyperref}
\usepackage{authblk}
\usepackage{geometry}
\usepackage{amsmath,amssymb,amsthm}
\usepackage[linesnumbered,ruled,vlined,noresetcount]{algorithm2e}
\DontPrintSemicolon

\SetCommentSty{mycommfont}
\SetArgSty{textnormal}
\SetKwInput{Notation}{Notation}
\SetKwInput{Predicates}{Predicates}
\SetKwProg{Fn}{function}{:}{end}
\SetKwInput{Variables}{Variables}
\SetKwInput{Parameter}{Parameter}
\SetKwInput{Note}{Note}
\SetKwInput{Remark}{Remark}
\SetKwInput{Initially}{Initially}

\usepackage{multirow}
\usepackage{xcolor}
\usepackage{xfrac}

\newtheorem{theorem}{Theorem}
\newtheorem{corollary}{Corollary}
\newtheorem{lemma}{Lemma}

\theoremstyle{definition}
\newtheorem{definition}{Definition}

\newtheorem{remark}{Remark}
\newtheorem{note}{Note}
\newtheorem{observation}{Observation}

\def\protocol#1{{\normalfont\textsc{#1}}\xspace}
\newcommand{\greedy}{\protocol{GreedyExp}}

\newcommand{\tif}{\textit{if\ \ }}

\newcommand{\calG}{\mathcal{G}}

\newcommand{\calE}{\mathcal{E}}
\newcommand{\calA}{\mathcal{A}}

\newcommand{\nats}{\mathbb{N}}
\newcommand{\natsp}{\mathbb{N}_{\ge 1}}

\newcommand{\vmap}{V_{\mathrm{map}}}
\newcommand{\vvis}{V_{\mathrm{vis}}}
\newcommand{\vdone}{V_{\mathrm{done}}}
\newcommand{\emap}{E_{\mathrm{map}}}
\newcommand{\timer}{\mathtt{timer}}

\newcommand{\nxt}{\mathtt{next}}

\newcommand{\vobs}{V_{\mathrm{obs}}}

\newcommand{\edel}{E_{\mathrm{del}}}
\newcommand{\pdel}{P_{\mathrm{del}}}
\newcommand{\popen}{P_{\mathrm{open}}}
\newcommand{\gmap}{G_{\mathrm{map}}}
\newcommand{\lmap}{\Lambda_{\mathrm{map}}}

\newcommand{\vtarget}{V_{\mathrm{target}}}

\newcommand{\dcur}{d_{\mathrm{cur}}}

\newcommand{\tinf}{T_{\cap}}

\newcommand{\vbl}{V_{\mathrm{blk}}}

\newcommand{\vcur}{v_{\mathrm{cur}}}

\newcommand{\id}{\mathit{id}}
\newcommand{\pin}{p_{\mathrm{in}}}

\newcommand{\almcon}{\xi}

\title{
  Online Graph Exploration under \texorpdfstring{$T$}{T}-Interval Connectivity
}
\date{}

\author[1]{Yuichi Sudo}
\author[2]{Naoki Kitamura}
\author[3]{Masahiro Shibata}
\author[4]{Junya Nakamura}
\author[5]{Sébastien Tixeuil}
\author[6]{Toshimitsu Masuzawa}
\author[1]{Koichi Wada}

\affil[1]{Hosei University, Tokyo, Japan}
\affil[2]{The University of Osaka, Osaka, Japan}
\affil[3]{Kyushu Institute of Technology, Fukuoka, Japan}
\affil[4]{Toyohashi University of Technology, Aichi, Japan}
\affil[5]{Sorbonne University, Paris, France}
\affil[6]{Notre Dame Seishin University, Okayama, Japan}

\begin{document}

\maketitle

\begin{abstract}
  We study deterministic exploration by a single agent in $T$-interval-connected graphs, a standard model of dynamic networks in which, for every time window of length $T$, the intersection of the graphs within the window is connected.
  The agent does not know the window size $T$, nor the number of nodes $n$ or edges $m$, and must visit all nodes of the graph.
  We consider two visibility models: in $KT_0$, the agent observes only its current node's identifier and port information; in $KT_1$, it additionally observes the identifiers of its currently available neighbors.
  We investigate two fundamental questions: the minimum window size that guarantees exploration, and the optimal exploration time under sufficiently large window size.

  For both models, we show that a window size $T = \Omega(m)$ is necessary, and we present deterministic algorithms whose required window size is $O(\almcon(n,m)\cdot m + n \log^2 n)$, where $\almcon(n,m) = \frac{\ln n}{1 + \ln m - \ln n}$.
  This window-size bound is tight for a wide range of $m$, in particular when $m = n^{1+\Theta(1)}$.
  Separately, for exploration time, we prove lower bounds of $\Omega((m - n + 1)n)$ in the $KT_0$ model and $\Omega(m)$ in the $KT_1$ model, even when the dynamic graph is $\infty$-interval-connected (i.e., even when all snapshots share a common spanning tree), and our algorithms match these time bounds up to a polylogarithmic factor. These bounds are also tight when parameterized solely by $n$: $\Theta(n^3)$ for $KT_0$ and $\Theta(n^2)$ for $KT_1$.
\end{abstract}

\thispagestyle{empty}
\addtocounter{page}{-1}

\clearpage

\section{Introduction}
\label{sec:introduction}
$T$-interval connectivity was introduced by Kuhn, Lynch, and Oshman in their STOC 2010 paper on distributed computation in dynamic networks~\cite{KLO10}.
It requires every window of $T$ consecutive time steps to have a connected intersection graph, while allowing the network to change indefinitely.
In that seminal paper, they observed that $1$-interval connectivity alone does not guarantee that a single mobile entity, called an \emph{agent}, can explore the graph (i.e., visit all nodes) without knowledge of future edge availability.
Even though the graph remains connected at every time step, an adversary can confine the agent to just two nodes forever.
This impossibility raises a natural question: \emph{how large must the window size $T$ be to make exploration possible?}
We investigate this question for deterministic exploration of unknown graphs by a single agent using only local observations.

We consider two visibility models: in $KT_0$, the agent observes only the identifier and port information at its current node; in $KT_1$, it additionally observes the identifiers of its currently available neighbors.
For both models, we identify, up to a small density-dependent factor, the minimum window size that enables deterministic single-agent exploration in arbitrary-topology graphs.
Moreover, for a broad range of graph densities, our algorithms achieve asymptotically optimal exploration time, matching the corresponding lower bounds.
Although the proposed algorithms are based on simple greedy strategies, their analysis is considerably more involved.
In particular, by combining extremal bounds on graph girth with a novel amortized analysis,
we derive those tight bounds on the minimum required window size as well as the exploration time.

In this section, we review related work, formalize our problem setting, and present our contributions in Sections \ref{sec:related_work}, \ref{sec:models}, and \ref{sec:contribution}, respectively.

\subsection{Related Work}
\label{sec:related_work}

Graph exploration by a single agent in an unknown undirected graph is a classical problem in theoretical computer science. The goal is to visit every node of the graph starting from an arbitrary node. Exploration algorithms often serve as foundational tools for solving other fundamental problems, such as rendezvous~\cite{Pelc12,TZ14,PP24}, gathering~\cite{DPP14,DFP+20,SKY+20,SKE+25,BDP23}, dispersion~\cite{AM18,SSKM20,SSN+26,KS25,KKM+25}, and gossiping~\cite{MT10,BDP23}.

This problem has been studied extensively in the standard port-numbering model, where edges have local \emph{port numbers} at each endpoint. For labeled nodes, depth-first search uses at most $2m$ edge traversals on a connected graph; Panaite and Pelc~\cite{PP99} reduced this bound to $m+3n$, where $n$ and $m$ count nodes and edges.
For anonymous nodes, exploration is possible using a Universal Traversal Sequence (UTS)~\cite{AKJ+79} or Universal Exploration Sequence (UXS)~\cite{koucky02,Reingold08,TZ14,Xin07} if an upper bound on $n$ is known, or using whiteboards (local node memory)~\cite{PDD+96,YWI+03,MT10,SBN15,SOK25}.
The rotor-router of Priezzhev, Dhar, Dhar, and Krishnamurthy~\cite{PDD+96} is self-stabilizing: using whiteboards, it repeatedly visits every node of any connected graph from any initial agent location and whiteboard states.
Its cover time (the number of edge traversals until all nodes are visited) is $O(mD)$~\cite{YWI+03}, where $D$ is the diameter.
Sudo, Ooshita, and Kamei~\cite{SOK25} studied time--space trade-offs for deterministic and randomized self-stabilizing exploration.

All the above results assume that the underlying graph is static. In many modern systems, however, the network topology evolves over time due to mobility, failures, or intermittent connectivity. Such systems are commonly modeled as \emph{time-varying graphs} or \emph{temporal graphs}, where the node set is fixed and the edge set changes over discrete time steps. Casteigts, Flocchini, Quattrociocchi, and Santoro~\cite{CFQS12} introduced a unifying framework for time-varying graphs and organized temporal connectivity assumptions into a hierarchy that has been influential in distributed computing, and Michail~\cite{Mic16} surveys algorithmic aspects of temporal graphs.

In dynamic graphs, exploration can be substantially harder and may even be impossible without additional assumptions on the dynamics. One line of research considers \emph{offline} temporal-graph problems where the entire evolution is given as input and the goal is to compute an \emph{exploration schedule} (a temporal walk) that minimizes the exploration time. Michail and Spirakis~\cite{MS16} introduced the \emph{temporal graph exploration (TEXP)} problem, and Erlebach, Hoffmann, and Kammer~\cite{EHK21} studied it extensively, establishing strong inapproximability bounds and presenting algorithms for special graph classes.

In the \emph{online} setting, future edge availability is unknown. As noted above, an adversary can prevent a single agent from exploring even when every snapshot is connected~\cite{KLO10}, so we require additional assumptions on the dynamics to make exploration possible.
Flocchini, Mans, and Santoro~\cite{FMS13} studied carrier-based graphs whose carriers follow periodic routes.
For a known underlying graph and known edge-specific recurrence bounds or appearance probabilities satisfying suitable cut conditions,~\cite{EHK21} gives exploration-time bounds of $O(m)$ for regularly present edges and $O(m\log n)$ with high probability for probabilistically present edges.
Gotoh, Sudo, Ooshita, and Masuzawa~\cite{GSOM20} studied exploration with partial predictions of future edge availability.

More recently, Baguley, G{\"o}bel, Klodt, Skretas, Sylvester, and Zamaraev~\cite{BGK+26} studied random spanning-tree models, in which each snapshot is sampled independently from a fixed distribution over spanning trees.
They proved high-probability offline exploration bounds of $O(n^{3/2})$ and $O(m)$, where $m$ is the number of edges in the underlying graph, and showed that the former bound is tight.
They also exhibited models in which every randomized online algorithm requires $\Omega(n^2)$ time with high probability, demonstrating a separation between offline and online exploration.

Exploration has also been studied under $T$-interval connectivity, primarily for restricted topologies. Ilcinkas and Wade~\cite{IW18} studied exploration of $T$-interval-connected dynamic rings and obtained tight bounds on the worst-case exploration time. In their offline setting, the agent knows the entire evolution of the graph, while in the online setting they additionally require $\delta$-recurrence, meaning that every edge appears at least once in every $\delta$ consecutive time steps, to guarantee that all nodes can be visited.
Beyond rings, Ilcinkas and Wade~\cite{IW21} studied exploration of
$1$-interval-connected cacti graphs (i.e., $T=1$) in the offline setting, giving
an $O(\frac{n\log n}{\log \log n})$-time algorithm.

For online exploration with multiple agents, Di Luna, Dobrev, Flocchini, and Santoro~\cite{DDFS20} studied agent requirements and cover time in 1-interval-connected rings, while Gotoh, Sudo, Ooshita, Kakugawa, and Masuzawa~\cite{GSO+21} investigated dynamic tori whose rows and columns are 1-interval-connected rings. Gotoh, Flocchini, Masuzawa, and Santoro~\cite{GFMS21} established tight bounds on the number of agents needed to explore temporal graphs of arbitrary topology in the online setting under temporal connectivity and stronger connectivity assumptions.

Exploration has also been studied in models that allow neighboring network nodes to exchange information and update their states in parallel, a capability absent from the standard mobile-agent model considered here.
Pritchard~\cite{Pri07} studied exploration under permanent edge deletions in such a model.
Nevertheless, his $\Omega(n^2)$ lower bound for arbitrary exploration strategies also applies to both $KT_0$ and $KT_1$, even under $\infty$-interval connectivity.
The edge-deletion strategy used inside our lower-bound gadget is essentially the same; Section~\ref{sec:lower} explains how the gadget's two nonadjacent gates enable its use in the window-size lower bound.

Despite this progress, the minimum window size and the optimal exploration time for deterministic online exploration by a single agent
on $T$-interval-connected graphs with arbitrary topology remain largely open.

\subsection{Models and Notations}
\label{sec:models}
Let $\nats = \{0,1,\dots\}$ denote the set of non-negative integers, and let
$\nats_{\ge i} = \{k \in \nats \mid k \ge i\}$ for any integer $i$.
For real numbers $x$ and $y$, we define the integer interval
$[x..y] = \{k \in \nats \mid x \le k \le y\}$.
We use $\log$ to denote the base-$2$ logarithm, and $\ln$ to denote the natural logarithm.
For any graph $G = (V,E)$, we denote by $d_G(u,v)$ the distance between nodes $u$ and $v$.

We consider a \emph{dynamic graph} $\calG = (V,\calE)$ with an \emph{underlying graph} $G=(V,E)$, where $V$ is a common set of nodes and $\calE: \nats \to 2^E$ is a function such that $\calE(t)$ denotes the set of edges that appear at time $t \in \nats$.
For each $t \in \nats$, we denote the static snapshot of $\calG$ at time $t$ by $G_t = (V,\calE(t))$.
For $T \in \natsp$, $\calG$ is said to be \emph{$T$-interval-connected}
if $\left(V, \bigcap_{i \in [t..t+T-1]} \calE(i)\right)$ is connected for all $t \in \nats$.
We say that $\calG$ is \emph{$\infty$-interval-connected} if $\left(V,\bigcap_{t\in\nats}\calE(t)\right)$ is connected, equivalently, if all snapshots share a common spanning tree.
We say that an edge $e \in E$ \emph{appears} or \emph{is present} at time $t$ if $e \in \calE(t)$.
The number of nodes (resp.~edges) of $\calG$ refers to that of its underlying graph $G$, i.e., $|V|$ (resp.~$|E|$), which we often denote by $n$ (resp.~$m$).
We say that $u, v \in V$ are \emph{neighbors} if $\{u,v\} \in E$, and denote by
$N(v) = \{u \in V \mid \{u,v\} \in E\}$ the set of neighbors of $v$.
The degree of a node $v$ is $\delta(v) = |N(v)|$.

Each node $v \in V$ has a unique label $\id(v)$, and we use $v$ and $\id(v)$ interchangeably when clear from the context.
We use the standard port-numbering interface of online graph exploration on the underlying graph: each edge incident to $v$ is assigned a locally unique port number that is fixed over time.
Formally, for each $v \in V$, there is a bijection $\lambda_v : N(v) \to [1..\delta(v)]$, and for $p \in [1..\delta(v)]$ we denote by $N(v,p)$ the unique neighbor $u$ with $\lambda_v(u)=p$.

For each $t \in \nats$, we denote by $\nu(t) \in V$ the node at which the agent is located at time $t$. For $t \ge 1$, we define $\pin(t)=\lambda_{\nu(t)}(\nu(t-1))$, that is, $\pin(t)$ is the \emph{incoming port} through which the agent arrived at its current location $\nu(t)$.\footnote{
  We omit the \emph{stay} option for notational convenience.
  If the agent waits, the adversary may repeat the current snapshot.
  Repeating snapshots preserves $T$-interval connectivity: every future sequence of snapshots that was possible before the repetition is still possible afterward.
  Thus waiting need not yield new information or restrict the adversary's future choices.
}
We define $\pin(0)=\bot$.
For $t \in \nats$ and $v \in V$, we define
$P(v,t)=\{\lambda_v(u) \mid \{u,v\} \in \calE(t)\}$,
that is, the set of available ports at $v$ at time $t$.

We consider two visibility models, $KT_0$ and $KT_1$, depending on whether the agent can observe the identifiers of its currently available neighboring nodes.
In the $KT_0$ model, at each time $t \in \nats$, the agent receives as input $(\id(\nu(t)), P(\nu(t),t), \pin(t))$ and selects a port $p \in P(\nu(t),t)$, moving to the neighbor $N(\nu(t),p)$. Alternatively, it may choose to \emph{terminate}.
In the $KT_1$ model, the agent additionally learns $\vobs(t) = \{ v \in V \mid \{\nu(t), v\} \in \calE(t)\}$ and $\lambda_{\nu(t)}(v)$ for each $v \in \vobs(t)$, i.e., the identifiers of all neighbors of $\nu(t)$ in $G_t$ together with the corresponding port numbers.
The agent does not know the window size $T$, the number of nodes $n$, or the number of edges $m$ a priori.

This distinction between $KT_0$ and $KT_1$ isolates the effect of one-hop identifier visibility under the classical fixed-port interface.
Although allowing port labels to change may provide a more natural model of dynamic graphs, it fundamentally changes the $KT_0$ model and makes exploration impossible even on a static graph.
If the agent is not told the incoming edge's new port number after reassignment, the adversary can predict its deterministic choice and assign the chosen port to the edge leading back to the previous node, confining it to two nodes.
Even if the agent learns the incoming edge's new port number, a similar strategy confines it to three nodes of a static $K_4$: the incoming port leads back to the previous node, and any other port selected by the agent can be assigned to the other node of the triangle.
Since the agent is deterministic, these assignments can be fixed before each observation.\footnote{
  By contrast, in the $KT_1$ model, dynamically changing port labels are computationally equivalent to the fixed-port interface, since at each time $t$ the agent learns $\lambda_{\nu(t)}(v)$ for every $v \in \vobs(t)$.
}
We therefore restrict our attention to the classical fixed-port interface throughout the paper to isolate the effect of one-hop identifier visibility.

\subsection{Our Contribution}
\label{sec:contribution}
This paper addresses the following two natural questions for single-agent exploration on $T$-interval-connected graphs:
\begin{itemize}
  \item What is the minimum window size $T$ that enables the exploration of any $T$-interval-connected graph with $n$ nodes (and $m$ edges)?
  \item Given a sufficiently large $T$, what is the optimal exploration time?
\end{itemize}

To formalize the first question, we define the functions $\psi_k(n)$ and $\psi_k(n,m)$ for the $KT_k$ model, where $k \in \{0,1\}$, as follows:
\begin{definition}
  For $n \ge 3$ and $m \in [n..\binom{n}{2}]$, let $\psi_k(n,m)$ denote the minimum $T$ such that there exists a deterministic algorithm under which a single agent explores every $T$-interval-connected graph with $n$ nodes and $m$ edges in the $KT_k$ model.
  We define
  $\psi_k(n) = \max \{ \psi_k(n,m) \mid m \in [n..\binom{n}{2}] \}$.
  Since additional edges can be made invisible to the agent, $\psi_k(n,m) \le \psi_k(n,m')$ whenever $m \le m'$. Hence, equivalently, $\psi_k(n)=\psi_k \left(n,\binom{n}{2}\right)$.

\end{definition}

\begin{note}
  We assume $m \ge n$, although the underlying graph may be connected even when $m = n-1$.
  We impose this assumption because we study exploration on $T$-interval-connected \emph{dynamic} graphs:
  if $m = n-1$, every edge must be present at all times to satisfy $T \ge 1$, and hence the problem reduces to the static case.
  This assumption also implies $n \ge 3$ for simple graphs.
\end{note}

Throughout this paper, for any $n \ge 3$ and $m \ge n$, we define
\[
  \almcon(n,m) = \frac{\ln n}{1 + \ln m - \ln n}
\]
which is chosen so that $n^{1 + 1/\almcon(n,m)} = e m$.
When $n$ and $m$ are clear from the context, we simply write $\almcon$ for $\almcon(n,m)$.
Over the range $n \le m \le \binom{n}{2}$, it lies between a positive constant and $\ln n$.

To answer the above questions, we first present the following upper bounds.

\begin{theorem}
  \label{thm:kt_zero_upper}
  In the $KT_0$ model, there exist a function $\tau(n,m) = O(\almcon(n,m)\cdot m + n \log^2 n)$ and a deterministic algorithm $\calA$ such that, for any $n \ge 3$ and $m\in[n..\binom{n}{2}]$, a single agent running $\calA$ explores any $\tau(n,m)$-interval-connected graph with $n$ nodes and $m$ edges in $O((m-n+1)n+n\log^2 n)$ time.
\end{theorem}

\begin{theorem}
  \label{thm:kt_one_upper}
  In the $KT_1$ model, there exist a function $\tau(n,m)=O(\almcon(n,m)\cdot m+n\log^2n)$ and a deterministic algorithm $\calA$ such that, for any $n \ge 3$ and $m\in[n..\binom{n}{2}]$, a single agent running $\calA$ explores
  any $\tau(n,m)$-interval-connected graph with $n$ nodes and $m$ edges
  in at most $\tau(n,m)$ time.
\end{theorem}

\begin{corollary}
  \label{cor:psi_nm_upper}
  For any $k \in \{0,1\}$, $\psi_k(n,m) = O(\almcon(n,m)\cdot m  + n \log^2 n)$.
\end{corollary}

The above upper bounds on the minimum window size are tight for a wide range of $m$, i.e., when $m = n^{1+\Omega(1)}$, since we prove the following lower bound.

\begin{theorem}
  \label{thm:psi_nm_lower}
  For any $k \in \{0,1\}$, $\psi_k(n,m) = \Omega(m)$.
\end{theorem}
Even if $m = n^{1+o(1)}$,
the gap between the above upper and lower bounds is at most
$\almcon(n,m) + (n \log^2 n)/m = O(\log^2 n)$.
Moreover, when parameterized solely by $n$, these bounds are tight:
we obtain the following corollary by substituting $m = \binom{n}{2}$.

\begin{corollary}
  \label{cor:psi_n}
  For any $k \in \{0,1\}$, $\psi_k(n) = \Theta(n^2)$.
\end{corollary}

\begin{table}
  \centering
  \caption{The required window size and exploration time}
  \label{tbl:results}
  \vspace{11pt}

  \begin{tabular}{cccc}
    \hline
                                                  & model          & required window size        & exploration time           \\
    \hline
    Theorem \ref{thm:kt_zero_upper} ($\greedy_0$) & $KT_0$         & $O(\almcon m + n \log^2 n)$ & $O((m-n+1)n+n\log^2 n)$    \\    
    Theorem~\ref{thm:kt_zero_lower} (lower bound) & $KT_0$         & -                           & $\Omega((m-n+1)n)$         \\
    \hline
    Theorem \ref{thm:kt_one_upper} ($\greedy_1$)  & $KT_1$         & $O(\almcon m + n \log^2 n)$ & $O(\almcon m + n\log^2 n)$ \\    
    Theorem~\ref{thm:kt_one_lower} (lower bound)  & $KT_1$         & -                           & $\Omega(m)$                \\
    \hline
    Theorem~\ref{thm:psi_nm_lower} (lower bound)  & $KT_0 \& KT_1$ & $\Omega(m)$                 & -                          \\
    \hline
  \end{tabular}
\end{table}

This paper also establishes the following lower bounds on the exploration time:

\begin{theorem}
  \label{thm:kt_zero_lower}  
  For every deterministic algorithm $\calA$ in the $KT_0$ model, every sufficiently large $n$, and every $m\in[n..\binom{n}{2}]$, there exist an $\infty$-interval-connected graph with $n$ nodes and $m$ edges and a starting node such that a single agent running $\calA$ from that node requires $\Omega((m-n+1)n)$ time to explore the graph.
\end{theorem}

\begin{theorem}
  \label{thm:kt_one_lower}  
  For every deterministic algorithm $\calA$ in the $KT_1$ model, every sufficiently large $n$, and every $m\in[n..\binom{n}{2}]$, there exist an $\infty$-interval-connected graph with $n$ nodes and $m$ edges and a starting node such that a single agent running $\calA$ from that node requires $\Omega(m)$ time to explore the graph.
\end{theorem}

Thus, we obtain a tight bound on the exploration time in the $KT_0$ model whenever $m = n + \Omega(\log^2 n)$.
The bound for the $KT_1$ model is also nearly tight, with the same multiplicative gap as that for the minimum window size.
Notably, these theorems establish tight bounds on the exploration time when parameterized solely by $n$: $\Theta(n^3)$ for $KT_0$ and $\Theta(n^2)$ for $KT_1$.

Our contributions are summarized in Table~\ref{tbl:results}.

\section{Upper Bounds}
\label{sec:upper}

In this section, we present two algorithms, $\greedy_1$ and $\greedy_0$, for the $KT_1$ and $KT_0$ models, respectively.
Throughout this section, we define
\[
  \tau(n,m) = c \cdot \lceil \almcon(n,m) \cdot m + n\ln^2 n \rceil,
\]
where $c \in \natsp$ is a constant that will be chosen sufficiently large.
In the remainder of this section, we fix a $\tau(n,m)$-interval-connected graph $\calG=(V,\calE)$ with underlying graph $G=(V,E)$, where $n = |V|$ and $m = |E|$.
We then show, for a sufficiently large constant $c$, that $\greedy_1$ explores $\calG$ within $\tau(n,m)$ time in the $KT_1$ model in Section~\ref{sec:1-hop}, whereas $\greedy_0$ explores $\calG$ within $O((m-n+1)n+n\log^2 n)$ time in the $KT_0$ model in Section~\ref{sec:0-hop}.
These results immediately prove Theorems~\ref{thm:kt_one_upper} and~\ref{thm:kt_zero_upper}, respectively.
In what follows, we denote by $\vcur \in V$ the current location of the agent, particularly in the pseudocode (Algorithms~\ref{al:greedy_one} and~\ref{al:greedy_zero}). In other words, $\vcur = \nu(t)$ at time $t$.

\subsection{With 1-hop view}
\label{sec:1-hop}

The strategy of $\greedy_1$, whose pseudocode is given in Algorithm~\ref{al:greedy_one}, is very simple.
During exploration, the agent constructs a map of the subgraph it has observed so far, removing all edges that it has observed to be unavailable at least once, and always moves toward the closest unvisited node in the map.
We describe this more precisely below.

The agent maintains two sets of nodes $\vmap, \vvis$
and the set of edges $\emap \subseteq \binom{\vmap}{2}$.
The set $\vmap$ (resp.~$\vvis$) is the set of nodes that the agent has observed (resp.~visited) so far. The set $\emap$ consists of the edges that the agent has observed at least once and that have been present on every visit to either endpoint since the start of the execution. Formally, at time $t$,
\begin{align*}
   & \vmap = \{\nu(0)\} \cup \vobs(0) \cup \vobs(1) \cup \cdots \cup \vobs(t), \qquad
  \vvis = \{\nu(s) \mid s \in [0..t]\},                                               \\
   &  
  \emap = \{\{v_0,v_1\} \mid (v_0 \in \vvis \lor v_1 \in \vvis) \land  (\forall s \in [0..t],\forall i \in \{0,1\}: \nu(s) = v_i \to v_{1-i} \in \vobs(s))\}.
\end{align*}
Note that $\vvis \subseteq \vmap$ always holds.
We then define the map $\gmap = (\vmap, \emap)$.

At each time, the agent chooses any node $w \in \vmap \setminus \vvis$ closest to $\vcur$ in $\gmap$ and moves to the next node on a shortest path from $\vcur$ to $w$ in $\gmap$ (line~2).
Note that the newly selected shortest path after the move may not be a suffix of the previous one, because the edge incident to the current node on that path may disappear at that time.
The agent terminates when $\vmap = \vvis$, i.e., when no unvisited nodes remain in the map (line~1).

\begin{algorithm}[t]
  \caption{$\greedy_1$. The update rules of $\vvis$, $\vmap$, and $\emap$ are omitted from the pseudocode; see the main text for these rules.
  }
  \label{al:greedy_one}  
  \While{
    $\vmap \setminus\vvis \neq \emptyset$
  }{
    Move to the next node on a shortest path from $\vcur$ to $w$ in $\gmap$,
    where $w$ is an arbitrary closest node in $\vmap \setminus \vvis$ from $\vcur$ in $\gmap$\;
  }
\end{algorithm}

In the remainder of Section~\ref{sec:1-hop}, we prove that under $\greedy_1$, a single agent visits all nodes of $\calG$ and terminates within $\tau(n,m)=O(\almcon m + n\log^2n)$ time.

\begin{remark}
  \label{remark:st}
  We first analyze $\greedy_1$ under the assumption that all snapshots
  share a common spanning tree $\tinf$.
  To justify this assumption, let $\tinf$ be a spanning tree common to
  the first $\tau(n,m)$ snapshots, and add its edges to every snapshot
  at time $t \ge \tau(n,m)$.
  This modification leaves the execution unchanged through time
  $\tau(n,m)-1$.
  The analysis below gives an $O(\almcon m+n\log^2 n)$ bound on the
  exploration time in the modified graph.
  For sufficiently large $c$, this bound is at most $\tau(n,m)-1$,
  so the same execution and termination occur in the original graph.
\end{remark}

\begin{lemma}
  \label{lem:correct}
  $\vmap = \vvis$ implies $\vvis = V$.
\end{lemma}

\begin{proof}
  Assume for contradiction that at some time $t$, we have $\vmap = \vvis \subsetneq V$.
  Let $E_T$ be the edge set of the spanning tree $\tinf$ mentioned in Remark~\ref{remark:st}.
  Then there exists an edge $\{u,v\} \in E_T$ with $u \in \vvis$ and $v \notin \vvis$.
  Since $u \in \vvis$ at time $t$, the agent has visited $u$ by time $t$, and hence $v \in \vmap$, because $\{u,v\} \in E_T$ and is therefore always present.
  This contradicts $\vmap = \vvis$.
\end{proof}

Under the assumption of Remark~\ref{remark:st}, $\greedy_1$ terminates in finite time.
Indeed, whenever an unvisited node remains, a path in $\tinf$ from $\vcur$ to its first unvisited node lies in $\gmap$, so an unvisited node is reachable.
Since $\vvis$, $\vmap$, and $\emap$ change only finitely many times, an infinite execution would eventually keep these sets fixed while the finite distance to the nearest unvisited node decreases at every move, a contradiction.
Lemma~\ref{lem:correct} then guarantees that all nodes have been visited at termination.

To bound the exploration time by $\tau(n,m)$, we introduce the potential
\[
  \dcur =
  \begin{cases}
    \min\{d_{\gmap}(\vcur,u) \mid u \in \vmap \setminus \vvis\} & \tif{\vvis \subsetneq \vmap}, \\
    0                                                           & \tif{\vvis = \vmap}.
  \end{cases}
\]
Since $\vcur \in \vvis$ at all times, $\dcur = 0$ if and only if $\vmap = \vvis$.
Thus, it suffices to show that $\dcur$ reaches zero within $\tau(n,m)$ time.

At any time, $\dcur$ is determined by $\vcur$, $\vvis$, $\vmap$, and $\emap$.
At time $0$, we have $\dcur = 1$ because $\vvis=\{\nu(0)\}$ and $\gmap$ is the star graph centered at $\nu(0)$ with leaf set $\vobs(0)$.
For the analysis, when the agent moves from $\nu(t)$ to $\nu(t+1)$ at time $t$,
we split the update of $\vcur$, $\vvis$, $\vmap$, and $\emap$ following each move into two stages.
\begin{enumerate}
  \item In the first stage, we update $\vcur$, $\vvis$, and $\vmap$,
        and insert all edges newly admitted to $\emap$, without removing any existing edge. Note that new edges are added only if the agent visits $\vcur$ for the first time.        
  \item In the second stage, we remove the edges excluded from $\emap$
        by the new observation, one at a time in an arbitrary fixed order.
\end{enumerate}
The intermediate map is used only for the analysis.

In the first stage, if $\nu(t+1)$ has already been visited, then $\dcur$ decreases by exactly one; otherwise, $\dcur$ may increase.
We define $\iota_V: V \to \nats \cup \{-1\}$ as follows: for any $v \in V \setminus \{\nu(0)\}$, $\iota_V(v)$ is the amount by which $\dcur$ increases during the first stage when the agent visits $v$ for the first time.
If $v$ is the last node visited by the agent, then $\dcur$ decreases from $1$ to $0$ at that time, and we define $\iota_V(v) = -1$ for this node. For simplicity, we set $\iota_V(\nu(0))=0$.
In the second stage, some edges incident to $\nu(t+1)$ may be eliminated from $\gmap$.
We fix an arbitrary order of these edges and assume that they are eliminated sequentially.
Each elimination may increase $\dcur$, but each edge $e \in E$ can be eliminated at most once throughout the execution of $\greedy_1$. To capture this increase, we define $\iota_E: E \to \nats$ as follows: for each edge $e \in E$, $\iota_E(e)$ is the amount by which $\dcur$ increases when $e$ is eliminated, and $\iota_E(e)=0$ if $e$ is never eliminated from $\gmap$.

\begin{lemma}
  \label{lem:potential}
  $\vmap = \vvis$ holds within at most $n+\sum_{v \in V}\iota_V(v) + \sum_{e \in E}\iota_E(e)$ time.
\end{lemma}

\begin{proof}
  Let $V_t$ denote $\vvis$ at time $t$, and let $\bar{E}_t$ be the set of edges eliminated from $\gmap$ when
  the agent moves from $\nu(t)$ to $\nu(t+1)$ in the second stage of time $t$.
  In the first stage of time $t$, $\dcur$ increases by $\iota_V(v)$ if the agent visited a new node (i.e., $|V_{t+1}| = |V_t|+1$); otherwise, $\dcur$ decreases by exactly one.
  In the second stage of time $t$, $\dcur$ increases by $\sum_{e \in \bar{E}_t} \iota_E(e)$.

  Since $\dcur=1$ at time $0$,
  we have
  \[
    \dcur
    = 1-(t-|V_t|+1)+\sum_{v \in V_t} \iota_V(v) + \sum_{j=0}^{t-1}\sum_{e \in \bar{E}_j} \iota_E(e)
    \le -t+n+\sum_{v \in V}\iota_V(v) + \sum_{e \in E}\iota_E(e)
  \]
  at time $t$.
  Since $\dcur$ is non-negative, it reaches zero within $n+\sum_{v \in V}\iota_V(v) + \sum_{e \in E}\iota_E(e)$ time, and $\dcur = 0$ holds only when $\vmap = \vvis$.
\end{proof}

By Remark~\ref{remark:st} and Lemmas~\ref{lem:correct} and~\ref{lem:potential}, it suffices to show $\sum_{v \in V}\iota_V(v) + \sum_{e \in E}\iota_E(e) = O(\almcon m + n \log^2 n)$.
We prove $\sum_{v \in V}\iota_V(v)\le 2 n\log n$ in Lemma~\ref{lem:nodes} and
$\sum_{e \in E}\iota_E(e)=O(\almcon m + n \log^2 n)$ in Lemma~\ref{lem:edges}.
We use the following lemma to prove Lemma~\ref{lem:nodes}.

\begin{lemma}
  \label{lem:tee_dist}
  Let $G' = (V',E')$ be any simple, undirected, and connected graph with $n' \ge 1$ nodes,
  and let $v_1, v_2, \dots, v_{n'}$ be the nodes of $G'$ in any order.
  Then, $\sum_{i \in [1..n'-1]} f(i) \le 2n' \log n'$,
  where $f(i) = \min \{ d_{G'}(v_i,v_j) \mid j \in [i+1..n'] \}$.
\end{lemma}

\begin{proof}
  We prove the lemma by induction on $n'$.
  The base case $n' = 1$ is trivial, as both sides are equal to $0$.
  Consider the general case $n' \ge 2$.
  Let $T_{G'}$ be any spanning tree of $G'$.
  Let $c_T$ be a centroid of $T_{G'}$, and let $T_1, T_2, \dots, T_{\ell}$ be the connected components of $T_{G'} - c_T$.\footnote{
    Every tree has one or two centroids.
  }
  For each $r \in [1..\ell]$, let $n_r = |V(T_r)|$.
  Since $c_T$ is a centroid, we have $n_r \le n'/2$ for every $r$, and $\sum_{r=1}^{\ell} n_r = n' - 1$.
  Fix $r$.
  Let $x_1, \dots, x_{n_r}$ be the nodes of $T_r$ in the order they appear in the sequence $v_1, v_2, \dots, v_{n'}$, and let $a_r = x_{n_r}$.
  For $s \in [1..n_r-1]$, define $f_r(s) = \min \{ d_{T_r}(x_s,x_t) \mid t \in [s+1..n_r] \}$.
  Then, for each $v_i = x_s \neq a_r$, we have $f(i) \le f_r(s)$.
  Therefore, by the induction hypothesis, the total contribution of all nodes in $T_r \setminus \{a_r\}$ is at most
  $\sum_{s=1}^{n_r-1} f_r(s) \le 2 n_r \log n_r \le 2n_r \log (n'/2) = 2n_r (\log n' - 1)$.
  We also need to bound the contribution of the exceptional nodes $a_1, a_2, \dots, a_{\ell}$ together with the centroid $c_T$.
  Let $b_1, b_2, \dots, b_{\ell+1}$ be these nodes in the order they appear in the sequence $v_1, v_2, \dots, v_{n'}$.
  Note that $b_{\ell+1} = v_{n'}$, and thus we ignore its contribution.
  Since $b_{i+1}$ appears later than $b_i$ in the sequence,
  the total contribution of the remaining nodes $b_1, b_2, \dots, b_{\ell}$ is bounded by
  $\sum_{i=1}^{\ell} d_{T_{G'}}(b_i,b_{i+1}) \le 2(n_1 + n_2 + \dots + n_{\ell}) \le 2n'$.
  Consequently,
  \[
    \textstyle
    \sum_{i=1}^{n'-1} f(i)
    \le \sum_{r=1}^{\ell} 2n_r (\log n' - 1) + 2n'
    \le 2n' (\log n' - 1) + 2n'
    \le 2n' \log n'. \qedhere
  \]
\end{proof}

\begin{lemma}
  \label{lem:nodes}
  $\sum_{v \in V}\iota_V(v) \le 2n \log n$.
\end{lemma}

\begin{proof}
  Let $v_1, v_2, \dots, v_n$ be the distinct nodes visited by the agent in this order,
  and let $\tinf$ be the spanning tree mentioned in Remark~\ref{remark:st}.
  By Lemma~\ref{lem:tee_dist},
  $\sum_{i=1}^{n-1} f(i) \le 2n \log n$,
  where $f(i) = \min \{ d_{\tinf}(v_i,v_j) \mid j \in [i+1..n] \}$.
  Since $\iota_V(v_n)=-1$, it suffices to show that $\iota_V(v_i) \le f(i)$.
  By definition, in $\tinf$, there exists a path $u_0, u_1, \dots, u_{f(i)}$ such that $u_0 = v_i$ and $u_{f(i)} = v_j$ for some $j \in [i+1..n]$.
  Since every edge in $\tinf$ is always present, there exists $k \in [1..f(i)]$ such that $\gmap$ contains the path $u_0, u_1, \dots, u_k$ and $u_k$ is unvisited at the time $v_i$ is first visited.
  This implies $\iota_V(v_i) \le k \le f(i)$.
\end{proof}

To bound $\sum_{e \in E}\iota_E(e)$ in Lemma~\ref{lem:edges},
for each $k \in \nats$, we define
\[
  F'(k) = \left|\{e \in E \mid \iota_E(e) \ge k\}\right|,
\]
that is, the number of edges whose deletion increases $\dcur$ by at least $k$.
We bound $F'(k)$ using the following well-known theorem from extremal graph theory
relating the girth and the number of edges.
Here, the \emph{girth} of a graph $G'$
is defined as the length of a shortest cycle in $G'$.
If $G'$ does not contain any cycle, its girth is defined as $\infty$.

\begin{theorem}[Theorem~4.1 in~\cite{FS13}\footnotemark{}]
  \label{thm:girth}
  Let $g \in \natsp$. For any graph $G'=(V',E')$ with girth at least $2g+1$,
  $
    |E'| \le \frac{1}{2}\,|V'|^{1+1/g} + \frac{1}{2}\,|V'|
  $.
\end{theorem}
\footnotetext{
  According to~\cite{FS13}, an essentially equivalent inequality was proved
  by Alon, Hoory, and Linial~\cite{AHL02}.
}

\begin{lemma}
  \label{lem:diff_edges}
  For any $k \in \natsp$, $F'(2k-1) \le \frac{1}{2}n^{1+1/k}-\frac{n}{2}+1$.
\end{lemma}

\begin{proof}
  Consider the subgraph $G_{X} = (V, X \cup E_T)$, where
  $X = \{e \in E \mid \iota_E(e) \ge 2k-1\}$,
  and $E_T$ is the edge set of the spanning tree $\tinf$ mentioned in Remark~\ref{remark:st}.
  By Theorem~\ref{thm:girth}, it suffices to show that $G_{X}$ has girth at least $2k+1$,
  since $F'(2k-1) = |X \cup E_T| - |E_T| \le \frac{1}{2}n^{1+1/k}-\frac{n}{2}+1$.
  Note that every edge in $E_T$ is always present, and thus $X \cap E_T = \emptyset$.

  Assume for contradiction that $G_{X}=(V,X\cup E_T)$ contains a cycle $C$ of length at most $2k$.
  Since $\tinf$ is a tree, $C$ contains at least one edge in $X$.
  Let $e=\{u,v\}$ be the first edge removed from $\gmap$ among those in $C$,
  and let $G'=(V,E')$ be $\gmap$ immediately before the removal of $e$.
  Then $d_{G'-e}(u,v) - d_{G'}(u,v) \ge 2k-1$,
  since otherwise the removal of $e$ would not increase $\dcur$ by at least $2k-1$.
  Thus, $d_{G'-e}(u,v) \ge 2k-1 + d_{G'}(u,v) = 2k$,
  implying that $G'$ does not contain the cycle $C$, i.e., $C \setminus E' \neq \emptyset$.

  Let $Y = C \setminus E'$, which is non-empty.
  Every edge in $E_T$ is always present, and $e$ is the first edge in $C$ removed from $\gmap$.
  Therefore, at the time immediately after the removal of $e$,
  (i) the agent is at either $u$ or $v$, i.e., $\vcur \in \{u,v\}$,
  (ii) both endpoints of every edge in $Y$ are unvisited, and
  (iii) all edges in $C \setminus (Y \cup \{e\})$ are in the current map $\gmap$.
  Hence, at this moment, there exists in $\gmap$ a path from $\vcur$ to some unvisited node consisting only of edges in $C \setminus (Y \cup \{e\})$,
  and thus $\dcur \le |C|-1 \le 2k-1$.
  Since $\dcur \ge 1$ before the removal of $e$, this contradicts $\iota_E(e) \ge 2k-1$.
\end{proof}

\begin{corollary}
  \label{cor:diff_edges}
  For any integer $k \ge 2\ln n$, $F'(2k-1) \le \frac{n\ln n}{k}+1$.
\end{corollary}

\begin{proof}
  For any real number $0 \le x < 1$, we have $e^x \le (1-x)^{-1}$.
  Thus, letting $y = (\ln n)/k \le 1/2$,
  $n^{1+1/k} = ne^{y} \le \frac{n}{1-y}
    = n + \frac{ny}{1-y} \le n + 2ny = n+ \frac{2n\ln n}{k}$.
  The corollary then follows from Lemma~\ref{lem:diff_edges}.
\end{proof}

\begin{lemma}
  \label{lem:edges}
  $\sum_{e \in E}\iota_E(e) = O(\almcon m+n \log^2 n)$.
\end{lemma}

\begin{proof}
  By definition, $F'(k) \le m$ for any $k \in \nats$.
  Moreover, since $\dcur$ is always at most $n-1$, $F'(n) = 0$.
  Therefore, we have
  \begin{align*}
    \sum_{e \in E}\iota_E(e)
     & \le \sum_{k \in [1..\frac{n}{2}]} 2k (F'(2k-1)-F'(2k+1))
    \le 2\sum_{k \in [1..\frac{n}{2}]} F'(2k-1)                                                \\
     & \le 2\ \left (
    \sum_{k \in [1..\almcon]} m
    +\sum_{k \in [\almcon..2\ln n]}    
    \frac{1}{2}n^{1+1/k}
    +\sum_{k \in [2\ln n..\frac{n}{2}]} \left (\frac{n \ln n}{k}+1 \right )
    \right)                                                                                    \\
     & = \sum_{k \in [\almcon .. 2\ln n]} n^{1+1/k}\quad   + \quad   O(\almcon m+ n \log^2 n).
  \end{align*}

  Thus, the first term is bounded by
  \[
    \sum_{k \in [\almcon..2\almcon]} n^{1+1/k} +\sum_{k \in [2\almcon..2\ln n]} n^{1+1/k}
    \le (\almcon+1) n^{1+1/\almcon} + 2n^{1+1/(2\almcon)}\ln n
    \le (\almcon+1) n^{1+1/\almcon} + 2\almcon n^{1+1/\almcon}
    = O(\almcon m),
  \]
  where in the second inequality of the displayed bound, we use the fact that, for any real number $x > 0$, $\ln n < x n^{1/(2x)}$.
  A proof of this fact is provided in the appendix for completeness (Lemma~\ref{lem:inequality}).
\end{proof}

Thus, Theorem \ref{thm:kt_one_upper} follows from Remark~\ref{remark:st} and Lemmas~\ref{lem:correct}, \ref{lem:potential}, \ref{lem:nodes}, and~\ref{lem:edges}.

\begin{algorithm}[t]
  \caption{$\greedy_0$. The updates of $\vmap$, $\vvis$, $\lmap$, $\pdel$, and $\delta'$ after each move are omitted; see the main text.}
  \label{al:greedy_zero}  
  \While{$\vtarget \neq \emptyset$}{
    $\timer \gets 0$\;
    $\vvis \gets \{\vcur\}$\;
    $\pdel(v) \gets \emptyset$ for all $v \in \vmap$\;
    $\pdel(\vcur) \gets \natsp \setminus P(\vcur,t)$, where $t$ is the current time\;
    \While{
      $\vtarget \neq \emptyset$ and $\timer + 1 < \tau(n',m')$
    }{
      \uIf{$\vcur \in \vtarget$}{
        Move via an arbitrary port in $\popen(\vcur) \setminus \pdel(\vcur)$\;
      }\Else{
        Move one step along a shortest path in $\gmap$ to a closest node in $\vtarget$\;
      }
      $\timer \gets \timer + 1$\;

    }
  }
\end{algorithm}

\subsection{With 0-hop view}
\label{sec:0-hop}
The strategy of $\greedy_1$ is to move toward a closest known but unvisited node.
In $KT_0$, however, the agent learns a node's identifier only upon visiting it, so every node in its map has already been visited.
Our algorithm $\greedy_0$ (Algorithm~\ref{al:greedy_zero}) keeps the nearest-target strategy but seeks unused ports at known nodes: traversing the corresponding edge is an attempt to discover a new node.

Traversing an unused edge may lead to an already visited node, so exploration may outlast the window $\tau(n,m)$ used in the analysis of $\greedy_1$.
The common spanning-tree assumption of Remark~\ref{remark:st} then need not hold for the entire execution, and permanent edge exclusions may obstruct exploration.
We therefore divide the execution into \emph{rounds}, retaining learned node and port information but resetting the records of unavailable ports.
An edge excluded in one round can thus be reconsidered in the next, while an edge already traversed is never again treated as unused.

To implement this strategy, the agent maintains the set $\vmap$ of all nodes visited so far and a table $\lmap$ of learned port assignments.
Traversing an edge from $u$ to $v$ reveals both $\lmap(u,v)=\lambda_u(v)$ and $\lmap(v,u)=\lambda_v(u)$ through the outgoing and incoming ports; unlearned entries equal $\bot$.
It also maintains $\delta'(v)$, the largest port number learned at $v$ from an available-port set or as an incoming port.
These records persist across rounds.
The \emph{open ports} at $v \in \vmap$ are the ports in $[1..\delta'(v)]$ whose edges have never been traversed:
\[
  \popen(v) = \{p \in [1..\delta'(v)] \mid \forall u \in \vmap:\ p \neq \lmap(v,u)\}.
\]

The agent records exclusions by port number, so it can remember an absence before identifying the corresponding edge.
For each $v \in \vmap$, the set $\pdel(v) \subseteq \natsp$ records the positive integers absent from at least one available-port set observed at $v$ in the current round.\footnote{
  Each $\pdel(v)$ is either empty or the complement of a finite set in $\natsp$, so it has a finite representation.
}
A new entry starts empty. At every visit to $v$ at time $t$, including the first, the agent applies
\[
  \pdel(v) \gets \pdel(v) \cup \bigl(\natsp \setminus P(v,t)\bigr).
\]
At the start of a round, all sets $\pdel(v)$ for $v \in \vmap$ are cleared, and the current node's observation is processed by the same update.
The map $\gmap=(\vmap,\emap)$ used for routing retains the previously traversed edges whose ports are not excluded at either endpoint:
\[
  \emap = \bigl\{\{v_0,v_1\} \in \tbinom{\vmap}{2} \mid
  \forall i \in \{0,1\}: \lmap(v_i,v_{1-i}) \notin \pdel(v_i) \cup \{\bot\} \bigr\}.
\]

Unknown degrees require one further refinement: at the start of a new round, a known node $v$ may have unseen ports above $\delta'(v)$ even if $\popen(v)=\emptyset$.
We therefore target both nodes with an open port not excluded in this round and known nodes not yet visited in this round.
Let $\vvis$ be the set of nodes visited in the current round, reset to $\{\vcur\}$ at its start.
The target set is
\[
  \vtarget = \{v \in \vvis \mid \popen(v) \setminus \pdel(v) \neq \emptyset\}
  \cup (\vmap \setminus \vvis).
\]

On each arrival, including the initial placement at time $0$, the agent updates its records and adds the current node to $\vvis$ before evaluating the loop conditions or choosing its next move.
If $\vtarget=\emptyset$, it terminates.
Otherwise, if $\vcur \in \vtarget$, it traverses an edge through an arbitrary port in $\popen(\vcur) \setminus \pdel(\vcur)$; if not, it moves one step along a shortest path in $\gmap$ to a closest target.
In particular, a node reached for the first time in a round may cease to be a target once its current ports are observed.
Lemma~\ref{lem:well-defined} guarantees that a target remains reachable whenever $\vtarget \neq \emptyset$.

Finally, the rounds must be short even though the agent does not know $n$ or $m$.
The variable $\timer$ counts moves within a round, and a new round starts when $\timer+1=\tau(n',m')$ while targets remain, where
\[
  \textstyle
  n' = \max\left\{3,|\vmap|,1+\max_{v\in\vmap}\delta'(v)\right\}, \qquad
  m' = \max\left(n',\left\lceil \frac{1}{2}\sum_{v \in \vmap}\delta'(v) \right\rceil\right)
\]
The additional term in $n'$ uses the fact that observing port $p$ implies that the graph has at least $p+1$ nodes.
The estimates are recomputed throughout the round and are nondecreasing.
Moreover, $n'\le m'\le\binom{n'}2$, since $n'\ge3$ and
$\sum_{v\in\vmap}\delta'(v)\le|\vmap|(n'-1)\le n'(n'-1)$.
Since $n' \le n$ and $m' \le m$, we have $\tau(n',m') \le \tau(n,m)$.\footnote{
  $\almcon(n,m)m = \frac{m \ln n}{1+\ln m-\ln n}$ is nondecreasing in both $n$ and $m$ for
  $n \ge 3$ and $m \ge n$; see Lemma~\ref{lem:epsilon_monotone}.
}
The extra $1$ in the timer test counts the initial observation: a round of $L$ moves uses the $L+1$ observations at its starting time and after each move.
Thus, each round uses at most $\tau(n,m)$ consecutive snapshots, including its final arrival, and admits a spanning tree present throughout these snapshots.
Consecutive rounds share their boundary observation; resetting the records requires no move.

\begin{lemma}
  \label{lem:well-defined}
  While $\vtarget \neq \emptyset$, there exists at least one node in $\vtarget$ in the connected component of $\gmap$ that contains $\vcur$. Moreover, the agent terminates only after it has visited all nodes in $V$.
\end{lemma}

\begin{proof}
  In each round $r$, there exists a spanning tree $T_r$ such that all its edges are present throughout the round. At any time during round $r$, one of the following holds:
  (i) $\gmap$ contains every edge of $T_r$, or
  (ii) some node incident with an unused edge in $T_r$ is reachable from $\vcur$ in $\gmap$.
  In case (ii), that node is a target: either it has not yet been visited in this round, or the port of the unused tree edge is open and has not been excluded.
  Therefore, if $\vtarget \neq \emptyset$, there exists a node in $\vtarget$ that is reachable from $\vcur$ in $\gmap$; otherwise, $\gmap$ contains all edges of $T_r$, implying that all nodes have already been visited.
\end{proof}

Let $n_r$ and $m_r$ be the values of $n'$ and $m'$ at the end of round $r$, respectively.
Each round $r$ lasts at most $\tau(n_r,m_r)-1$ time steps, with equality unless the agent terminates in that round.
Indeed, the threshold is integer-valued and nondecreasing, so the first failed timer test gives $\timer+1=\tau(n_r,m_r)$.
If the agent traverses an unused edge $e=\{u,v\}$ from $u$ to $v$, but the destination $v$ has already been visited, we call this movement a \emph{redundant movement}.
Such movements are inevitable in the $KT_0$ model, since the agent does not know the destinations of unused edges.
Let $z_r$ be the number of redundant movements made by the agent during round $r$. Then the following lemma holds.

\begin{lemma}
  \label{lem:open_decrease}
  If $\vtarget \neq \emptyset$ at the end of round $r$, then $z_r \ge \frac{\tau(n_r,m_r)}{2n}$.
\end{lemma}

\begin{proof}
  Fix a round $r$ that ends with $\vtarget\neq\emptyset$.
  Let $W$ be $\vmap$ at its end.
  We bound the increases of the potential
  \[
    \dcur' = \min\{d_{\gmap}(\vcur,u)\mid u\in\vtarget\},
  \]
  which lies in $[0..n_r-1]$ by Lemma~\ref{lem:well-defined}.

  Once a node $v$ has been visited in this round, $\popen(v)\setminus\pdel(v)$ is nonincreasing.
  Indeed, traversing an edge removes its port from $\popen(v)$, and $\pdel(v)$ only grows.
  If an increase in $\delta'(v)$ adds a port number to $\popen(v)$, that number was absent at the preceding visit to $v$ and hence already belongs to $\pdel(v)$.
  Thus, a node can leave $\vtarget$ at most once in the round, and never reenters it.
  Likewise, an edge removed from $\gmap$ cannot reenter it in the same round.

  As in the analysis of $\greedy_1$, we assign charges to nodes and edges to account for potential increases.
  Let $\iota'_V(v)$ and $\iota'_E(e)$ denote the charges assigned to a node $v$ and an edge $e$, respectively.
  We account for redundant movements separately: each starts with potential zero and increases it by at most $n_r-1$, including the arrival updates, to which we assign no charges.
  For each non-redundant move from $u$ to $v$, we distinguish three cases according to the status of $v$ before the move.
  \begin{enumerate}
    \item \emph{Known target.}
          The potential is zero upon arrival, before the observation.
          If the updates remove $v$ from $\vtarget$, set $\iota'_V(v)$ to the resulting potential, including the effects of edge deletions; otherwise, assign no charge.
          Assign no edge charges.

    \item \emph{Newly discovered node.}
          The potential is zero before the move because $u$ is a target.
          If $v$ is not a target after the arrival updates, set $\iota'_V(v)$ to the resulting potential; otherwise, assign no charge.
          Assign no edge charges.

    \item \emph{Known nontarget.}
          The departure node $u$ must also be a nontarget before the move,
          since otherwise the move would be redundant.
          Thus, the target set remains unchanged, and we assign charges only to edges.
          Delete the edges excluded by the observation at $v$ one at a time
          in an arbitrary fixed order, and define $\iota'_E(e)$ as the
          potential increase caused by deleting $e$.
          Every intermediate map contains the fully updated map,
          so the potential remains finite throughout.
  \end{enumerate}
  Each node and edge receives at most one charge; set all unassigned charges to zero, and write
  \[
    \textstyle
    I_V=\sum_{v\in W}\iota'_V(v),\qquad
    I_E=\sum_{e\in E}\iota'_E(e).
  \]

  Let $L$ be the number of moves in round $r$.
  Each traversed edge contributes a learned port at both endpoints, so at most $m_r$ distinct edges have been traversed by the end of the round.
  Hence at most $m_r$ moves in this round traverse unused edges; all other moves decrease $\dcur'$ by one before the arrival observation.
  Since the initial potential is at most $n_r-1$ and the final potential is nonnegative,
  \[
    L\le n_r-1+m_r+(n_r-1)z_r+I_V+I_E.
  \]
  As $L+1=\tau(n_r,m_r)$, it remains to bound $I_V+I_E$.

  The charge bounds follow from essentially the same arguments as in the $KT_1$ model (Lemmas~\ref{lem:nodes} and~\ref{lem:edges}), with the following minor modifications.
  For $I_V$, we order nodes by when they leave the target set, rather than by their first visits.
  For $I_E$, the endpoints of a missing edge on the cycle used in the girth argument may already have been visited; we instead show that the endpoint reached just before this edge is a target.
  The vertex argument treats each component of the common tree induced by the known nodes separately; the edge argument joins these components into a tree on at most $n_r$ nodes.
  These modifications yield
  $I_V=O(n_r\log n_r)$ and
  $I_E=O(\almcon(n_r,m_r)m_r+n_r\log^2n_r)$.
  For completeness, we provide the detailed proofs in Appendix~\ref{app:kt0_charges} (Lemmas~\ref{lem:kt0_nodes} and~\ref{lem:kt0_edges}).

  Since $n_r\le m_r\le\binom{n_r}2$, we have $\almcon(n_r,m_r)=\Omega(1)$.
  Choosing the constant $c$ in $\tau$ sufficiently large gives $n_r+m_r+I_V+I_E\le\tau(n_r,m_r)/2$.
  The accounting inequality therefore yields
  \[
    (n_r-1)z_r\ge\tau(n_r,m_r)-(n_r+m_r+I_V+I_E)
    \ge\frac{\tau(n_r,m_r)}2.
  \]
  The claim follows from $n_r\le n$.
\end{proof}

\begin{corollary}
  \label{cor:total_time}
  The agent terminates in $O((m-n+1)n+n\log^2 n)$ time.
\end{corollary}
\begin{proof}
  By definition, the agent makes at most $m-n+1$ redundant movements.
  Thus, by Lemma~\ref{lem:open_decrease}, the total time required for all rounds except the last is at most $2(m-n+1)n$. The last round also completes within $\tau(n,m)=O(\almcon m+n\log^2 n)=O((m-n+1)n+n\log^2 n)$ time by definition.
\end{proof}

Thus, Theorem~\ref{thm:kt_zero_upper} follows from Lemma~\ref{lem:well-defined} and Corollary~\ref{cor:total_time}.

\section{Lower Bounds}
\label{sec:lower}
In this section, we prove the lower bounds stated in Section~\ref{sec:contribution}.
Specifically, we prove Theorems~\ref{thm:psi_nm_lower},~\ref{thm:kt_zero_lower}, and~\ref{thm:kt_one_lower}.

\begin{figure}[t]
  \centering
  \includegraphics[width=0.5\linewidth]{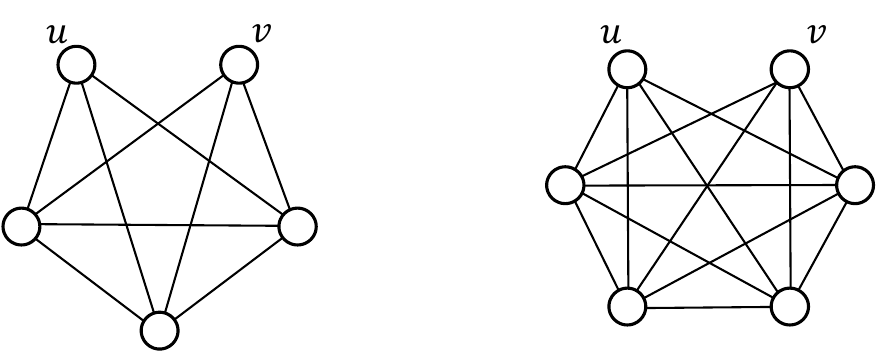}
  \caption{$K_5(u,v)$ (left) and $K_6(u,v)$ (right)}
  \label{fig:gadget}
\end{figure}

\begin{note}
  The proofs of Theorems~\ref{thm:psi_nm_lower} and~\ref{thm:kt_one_lower} share a gadget that confines the agent to a $k$-node subgraph for $\Omega(k^2)$ time (Lemma~\ref{lem:clique}).
  The edge-deletion strategy used inside the gadget is essentially the same as that underlying the $\Omega(n^2)$ exploration-time lower bound in~\cite{Pri07}.
  Our gadget has two designated, nonadjacent gates: any passage through the gadget from one gate to the other must enter its interior, where the confinement strategy applies.
  For Theorem~\ref{thm:kt_one_lower}, choosing the gadget size according to $m$ yields the exploration-time lower bound.
  The window-size lower bound in Theorem~\ref{thm:psi_nm_lower} requires an additional construction.
  We connect two copies and use them alternately, keeping one fully restored while the agent is confined in the other, to prevent exploration indefinitely while preserving $T$-interval connectivity.
\end{note}

We define the shared gadget $K_n(u,v)$ by removing the edge $\{u,v\}$ from the complete graph on $n$ nodes (see Figure~\ref{fig:gadget} for examples).
We call $u$ and $v$ the \emph{gates} of $K_n(u,v)$, and the remaining nodes \emph{traps}.
The key advantage of this gadget is that once the agent visits any trap, the adversary can confine the agent within the gadget for $\Omega(n^2)$ time, as formalized in the following lemma.

\begin{lemma}
  \label{lem:clique}
  For any deterministic algorithm $\calA$ and any sufficiently large $n$, there exists an $\infty$-interval-connected graph $\calG$ with underlying graph $K_n(u,v) = (V,E)$ such that a single agent running $\calA$ in the $KT_1$ model, starting from a node other than $u$ or $v$, requires $\Omega(n^2)$ time to reach either $u$ or $v$.
  This claim holds even if the agent knows the entire structure of the underlying graph, including the port assignments $\lambda_x(y)$ for all $x, y \in V$, a priori.
\end{lemma}

\begin{proof}
  Note that $\calG$ can depend on algorithm $\calA$, which is deterministic.
  Therefore, it suffices to describe an adversarial strategy that adaptively determines, at each time $t$, which edges are present based on $\nu(t)$ (i.e., the current location of the agent), thereby preventing the agent from reaching a designated node $u$ or $v$ for $\Omega(n^2)$ time while preserving $\infty$-interval connectivity.

  The adversary maintains a set $\vbl \subseteq V$ of \emph{blocked} nodes, a set $\vdone \subseteq V$ of nodes visited by the agent since the last update of $\vbl$, and a set $\edel$ of \emph{deleted} edges.
  The set $\edel$ represents the output of the adversarial strategy.
  That is, at each time step $t$, the adversary defines $\calE(t) = E \setminus \edel$.
  Initially, $\vbl = \{u,v\}$, $\vdone = \{\nu(0)\}$, and $\edel = \emptyset$.
  The adversary proceeds as follows:
  \begin{enumerate}
    \item $\edel \gets \edel \cup \{ \{\vcur,w\} \in E \mid w \in \vbl\}$.
    \item Let the agent make one move according to $\calA$.
    \item $\vdone \gets \vdone \cup \{\vcur\}$.
    \item If $|\vdone| + |\vbl| < n-1$, return to Step~1.
    \item $\vbl \gets \vbl \cup \{w\}$, where $w$ is the unique node in $V \setminus (\vbl\cup\vdone)$.
    \item $\vdone \gets \{\vcur\}$.
    \item If $|\vbl| < n-2$, return to Step~1; otherwise, terminate.
  \end{enumerate}

  If either the agent or the strategy terminates, the adversary keeps the current snapshot unchanged thereafter.
  The agent moves only in Step~2, whereas Step~1 ensures that all edges between the current location and the blocked nodes, including $u$ and $v$, belong to $\edel$, and hence are never present in Step~2.
  Therefore, the agent never visits $u$ or $v$ before this strategy terminates.
  The agent must make at least $n-2-|\vbl|$ moves to increase $|\vbl|$ by one.
  Thus, if the strategy finishes, the agent has made at least
  $\sum_{k=2}^{n-3} (n-2-k) = \Omega(n^2)$ moves; otherwise, it never reaches either gate.
  This lower bound is independent of the algorithm $\calA$.

  It remains to show that the strategy preserves $\infty$-interval connectivity.
  At every stage, the unblocked nodes induce a clique, and at least one of them has not yet been visited in that stage, by Step~4.
  If $w_1,\dots,w_k$ have been added to $\vbl$ in this order, with $k\ge 1$, the edges $\{u,w_1\}$, $\{v,w_1\}$, and $\{w_i,w_{i+1}\}$ for $i\in[1..k-1]$ remain present, and $w_k$ remains adjacent to such an unvisited node.
  If $k=0$, both gates remain adjacent to such a node instead.
  Hence every snapshot is connected.
  Since edges are only deleted and their number is finite, the intersection of all snapshots is connected as well, even if a stage never finishes.
\end{proof}

\begin{figure}[t]
  \centering
  \includegraphics[width=0.5\linewidth]{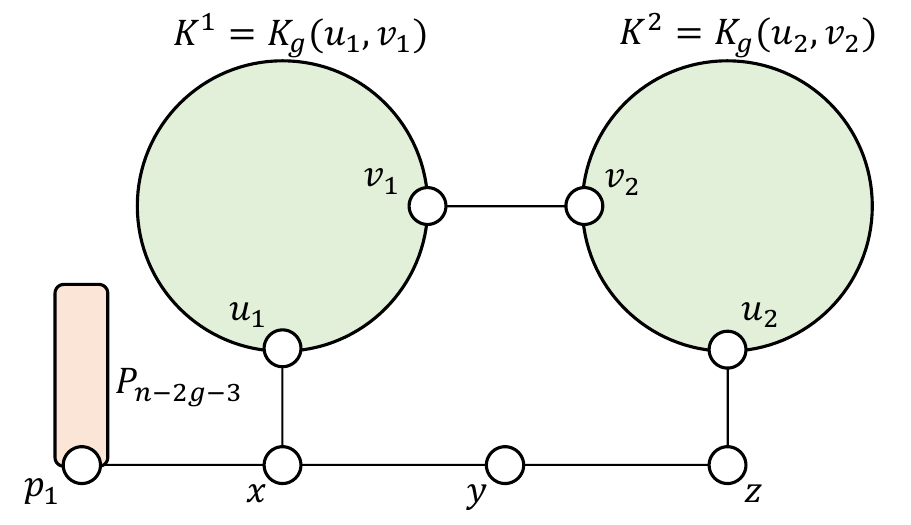}
  \caption{The underlying graph for the proof of Theorem~\ref{thm:psi_nm_lower}, excluding null edges}
  \label{fig:psi_nm_lower}
\end{figure}

Then, we prove Theorem~\ref{thm:psi_nm_lower}, which states that
$\psi_k(n,m) = \Omega(m)$ for any $k \in \{0,1\}$.

\begin{proof}[Proof of Theorem~\ref{thm:psi_nm_lower}]
  We prove that, for any sufficiently large $n$ and any $m \in [n..\binom{n}{2}]$, there exists an underlying graph $G=(V,E)$ with $n$ nodes and $m$ edges such that the adversary can prevent the agent from visiting all nodes while preserving $\lfloor c'm\rfloor$-interval connectivity for some constant $c'>0$.
  We prove this by considering two cases: $2n \le m \le \binom{n}{2}$ and $n \le m < 2n$.

  First, consider the case $2n \le m \le \binom{n}{2}$.
  Let $g = \lfloor\sqrt{m/8}\rfloor$.
  We define $G_0 = (V,E_0)$ as the graph obtained by merging two copies of the gadget, $K^1=K_g(u_1,v_1)$ and $K^2=K_g(u_2,v_2)$,
  a path graph $P_{n-2g-3} = (\{p_1,p_2,\dots,p_{n-2g-3}\},\{\{p_i,p_{i+1}\} \mid i \in [1..n-2g-4]\})$,
  and three nodes $x$, $y$, and $z$, connecting them with the following edges
  (see Figure~\ref{fig:psi_nm_lower}):
  \[
    \{u_1,x\}, \{x,y\}, \{y,z\}, \{z,u_2\}, \{v_1,v_2\}, \text{ and } \{x,p_1\}.
  \]
  (Note that $2g+3 \le n/2+3 < n$ for sufficiently large $n$.)
  Clearly, $|V| = n$. However,
  \[
    |E_0| = 2\left (\binom{g}{2} - 1\right)  + n - 2g - 4 + 6< g^2 + n \le \frac{m}{8} + n < m,
  \]
  for sufficiently large $n$.
  Therefore, to match the desired number of edges in $G$, we add $m - |E_0|$ \emph{null} edges to $G_0$ to obtain $G$, where a null edge is an edge that is never present when the agent visits its endpoints and is therefore never available to the agent.

  As in the proof of Lemma~\ref{lem:clique}, it suffices to describe an adaptive adversarial strategy that determines, at each time $t$, which edges are present based on $\nu(t)$, thereby preventing the agent from visiting all nodes while preserving $T$-interval connectivity, where $T = \lfloor c'm\rfloor$ for a sufficiently small constant $c'>0$.

  The adversary maintains a variable $\nxt \in \{1,2\}$ such that, outside confinement periods, $\nxt = i$ indicates that no edge of $K^i$ has been deleted during the past $T$ time steps.
  Thus, by Lemma~\ref{lem:clique}, once the agent visits any trap node of $K^{\nxt}$, the adversary can confine the agent within $K^{\nxt}$ for the next $T\le c'm \le 8c' (g+1)^2$ time steps while preserving $T$-interval connectivity.
  (Note that $c'$ can be an arbitrarily small constant.)
  Initially, $\nxt = 1$.

  The goal of the adversary is to prevent the agent from visiting $y$.
  We may choose any node other than $y$ as the initial location $\nu(0)$.
  Whenever the agent visits a trap node in $K^{\nxt}$, the adversary confines it within $K^{\nxt}$ for $T$ time steps.
  During this period, it does not delete any (non-null) edge outside $K^{\nxt}$, in particular, any edge of $K^{3-\nxt}$.
  At the end of each confinement period, it restores all edges of that gadget; all gadget edges remain present outside confinement periods.
  Thus, it can safely switch the value of $\nxt$ to $3-\nxt$ after the confinement period ends.
  In addition, the adversary deletes the edge $\{x,y\}$ (resp.~$\{y,z\}$) at time $t$ if and only if $\nu(t)=x$ (resp.~$\nu(t)=z$), so the agent can never visit $y$.
  This rule would violate $T$-interval connectivity if the agent visits both $x$ and $z$ within $T$ time steps, thereby isolating $y$.
  However, this does not occur, because whenever the agent moves from $x$ to $z$ or from $z$ to $x$, it must visit a trap node in $K^{\nxt}$, regardless of whether $\nxt = 1$ or $\nxt = 2$, and hence requires at least $T$ time steps to reach the destination.
  (Note that $\{u_1,v_1\}, \{u_2,v_2\} \notin E_0$.)

  To verify connectivity also across confinement boundaries, consider any window of $T$ consecutive snapshots.
  For each gadget, the window intersects at most one confinement period, since successive confinements in the same gadget are separated by a confinement of $T$ time steps in the other gadget.
  By Lemma~\ref{lem:clique} and the restoration rule, the intersection of these snapshots therefore contains a spanning tree of each gadget.
  Moreover, at least one of $\{x,y\}$ and $\{y,z\}$ is present throughout the window, since the agent cannot visit both $x$ and $z$ within it.
  All other edges connecting the gadgets, $x,y,z$, and the attached path are always present, so the intersection is connected.

  Thus, the agent can never visit $y$, while $T = \Omega(m)$-interval connectivity is preserved, which completes the proof, in the case $m \ge 2n$.

  Next, we consider the case $n \le m < 2n$, which is much simpler to handle.
  As an underlying graph, we consider a cycle $C_n=(\{v_0,v_1,\dots,v_{n-1}\},\{\{v_i,v_{i+1 \bmod{n}}\} \mid i \in [0..n-1]\})$,
  adding $m-n$ null edges.
  We may choose any node other than $v_1$ as $\nu(0)$.
  Then, the adversary can easily prevent the agent from visiting $v_1$:
  the edge $\{v_0,v_1\}$ (resp.~$\{v_1,v_2\}$) is absent exactly when the agent is at $v_0$ (resp.~$v_2$), and present otherwise; all other cycle edges remain present.
  This schedule preserves $T$-interval connectivity for $T = n - 3 > m/2 - 3 = \Omega(m)$, since the agent requires $n - 2$ steps to move from $v_0$ to $v_2$ (resp.~from $v_2$ to $v_0$), so once $\{v_0,v_1\}$ (resp.~$\{v_1,v_2\}$) is deleted, $\{v_1,v_2\}$ (resp.~$\{v_0,v_1\}$) remains present for the next $n - 3$ steps.
\end{proof}

We can also use the gadget $K_n(u,v)$ to prove Theorem~\ref{thm:kt_one_lower}: for every deterministic algorithm $\calA$ in the $KT_1$ model, every sufficiently large $n$, and every $m\in[n..\binom{n}{2}]$, there exist an $\infty$-interval-connected graph with $n$ nodes and $m$ edges and a starting node such that a single agent running $\calA$ from that node requires $\Omega(m)$ time to explore the graph.
The proof is similar to that of Lemma~\ref{lem:clique}, but more direct.

\begin{proof}[Proof of Theorem~\ref{thm:kt_one_lower}]
  If $m<2n$, the bound follows from the $n-1=\Omega(m)$ moves needed to visit all nodes of any connected static graph.
  Now assume $m\ge 2n$.
  We construct an underlying graph $G=(V,E)$ with $n$ nodes and $m$ edges.
  Let $g = \lfloor \sqrt{m} \rfloor < n$.
  We define $G_0 = (V,E_0)$ as the graph obtained by connecting a path graph $P_{n-g} = (\{p_1,p_2,\dots,p_{n-g}\},\{\{p_i,p_{i+1}\} \mid i \in [1..n-g-1]\} )$
  and $K_g(u,v)$ with an edge $\{p_1,u\}$.
  Clearly, $|V| = n$, and
  \[
    |E_0| = \binom{g}{2} - 1 + n - g \le m/2-1 + n \le m-1.
  \]
  Therefore, as in the proof of Theorem~\ref{thm:psi_nm_lower}, we add $m - |E_0|$ \emph{null} edges to $G_0$ to obtain $G$, which are never present whenever the agent visits their endpoints.
  Then, by Lemma~\ref{lem:clique}, we can prevent the agent from visiting either gate $u$ or $v$ for $\Omega(g^2) = \Omega(m)$ time while preserving $\infty$-interval connectivity, provided that the agent starts from a trap node in $K_g(u,v)$.
\end{proof}

\begin{figure}[t]
  \centering
  \includegraphics[width=0.8\linewidth]{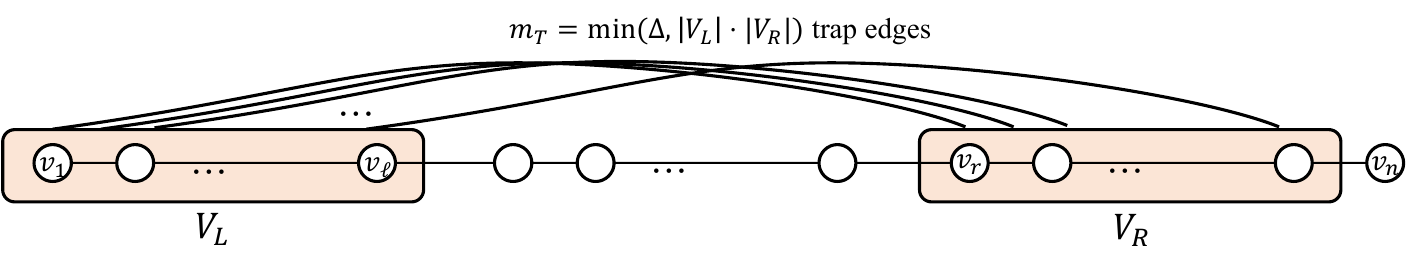}
  \caption{The underlying graph for the proof of Theorem~\ref{thm:kt_zero_lower}, excluding null edges}
  \label{fig:kt_zero_lower}
\end{figure}

Lastly, we prove Theorem~\ref{thm:kt_zero_lower}: for every deterministic algorithm $\calA$ in the $KT_0$ model, every sufficiently large $n$, and every $m\in[n..\binom{n}{2}]$, there exist an $\infty$-interval-connected graph with $n$ nodes and $m$ edges and a starting node such that a single agent running $\calA$ from that node requires $\Omega((m-n+1)n)$ time to explore the graph.

\begin{proof}[Proof of Theorem \ref{thm:kt_zero_lower}]
  Given any sufficiently large $n$ and any $m \in [n .. \binom{n}{2}]$,
  we construct an underlying graph $G = (V,E)$ with $n$ nodes and $m$ edges.
  Let $\Delta = m - n + 1$, $\ell = \lfloor n/3 \rfloor$, and $r = \lfloor 2n/3 \rfloor$.
  The node set is $V = \{v_1, v_2, \dots, v_n\}$.
  Let $V_L = \{v_1, v_2, \dots, v_{\ell}\}$ and $V_R = \{v_r, v_{r+1}, \dots, v_{n-1}\}$.
  Note that the rightmost node $v_n$ does not belong to $V_R$.
  The underlying graph $G$ has three types of edges: \emph{path edges}, \emph{trap edges}, and \emph{null edges} (Figure \ref{fig:kt_zero_lower}).
  The path edges form a path on the $n$ nodes; that is, $G$ contains the edge $\{v_i, v_{i+1}\}$ for all $i \in [1..n-1]$.
  Let $m_T = \min(\Delta, |V_L| \cdot |V_R|)$.
  We choose any $m_T$ trap edges, each connecting a node in $V_L$ and a node in $V_R$.
  If $m_T < \Delta$, we add $\Delta - m_T$ null edges to ensure $|E| = m$.

  As in the proof of Lemma~\ref{lem:clique}, it suffices to describe an adaptive adversarial strategy that determines, at each time $t$, which edges are present based on $\nu(t)$, thereby preventing the agent from visiting all nodes while preserving $\infty$-interval connectivity. Moreover,
  whenever the agent located at a node $v$ uses an \emph{unused} port $p \in [1..\delta(v)]$ to move, the adversary can decide, in a delayed manner, which \emph{unused} edge incident to $v$ is associated with $p$,
  since we assume the $KT_0$ model and we consider only deterministic algorithms.

  We describe a simple adversarial strategy. First, as in the proofs of the above lemmas, whenever the agent visits a node $v$, we delete all null edges incident to $v$. Next, we never remove any path edge at any time. Therefore, $\infty$-interval connectivity is preserved regardless of how we delete the trap edges.
  We choose $v_1$ as the starting node, i.e., $\nu(0) = v_1$, and prevent the agent from visiting $v_n$ for $\Omega(m_T \cdot n)=\Omega((m-n+1)n)$ time steps. Each trap edge is absent exactly when the agent is at its endpoint in $V_L$, and present otherwise. Therefore, once the agent visits a node in $V_L$, it requires $r - \ell = \Omega(n)$ time to visit any node in $V_R$.
  Moreover, whenever the agent attempts to traverse an unused edge at a node $v \in V_R$ (i.e., it chooses an unused port at $v$), we let the agent traverse an unused trap edge as long as such an edge exists. As a result, for each $v_i \in V_R$, the agent requires $\Omega(k_i \cdot n)$ time to visit $v_{i+1}$ after it first visits $v_i$, where $k_i$ is the number of trap edges incident to $v_i$.
  Therefore, the agent requires
  $\sum_{i \in [r..n-1]} \Omega(k_i \cdot n) = \Omega(m_T \cdot n) = \Omega(\min(\Delta, n^2)\cdot n) = \Omega((m - n + 1)n)$ time to visit $v_n$.
\end{proof}

\section{Conclusion}
\label{sec:conclusion}

In this paper, we studied deterministic single-agent exploration in $T$-interval-connected graphs under two visibility models, $KT_0$ and $KT_1$.
We focused on two natural questions: the minimum window size that guarantees exploration, and the exploration time achievable once the window size is sufficiently large.

We first established upper bounds for both models.
In particular, we presented deterministic exploration algorithms whose required window size is
$O(\almcon(n,m)\cdot m + n \log^2 n)$.
For the $KT_0$ model, our algorithm explores the graph in
$O((m-n+1)n + n\log^2 n)$ time,
while for the $KT_1$ model, the exploration time is at most the required window size itself.
We then proved matching or nearly matching lower bounds.
For both $KT_0$ and $KT_1$, we showed that the minimum required window size is $\Omega(m)$.
For the exploration time, we proved lower bounds of $\Omega((m-n+1)n)$ in the $KT_0$ model and $\Omega(m)$ in the $KT_1$ model.

These results yield a nearly complete picture of deterministic single-agent exploration under $T$-interval connectivity.
In particular, when $m = n^{1+\Theta(1)}$, our upper and lower bounds on the minimum window size match asymptotically, giving a tight bound $\Theta(m)$.
Moreover, when parameterized solely by $n$, our bounds are tight: $\psi_k(n)=\Theta(n^2)$ for both $k \in \{0,1\}$.
For the exploration time, our results also imply tight bounds when parameterized only by $n$, namely $\Theta(n^3)$ for $KT_0$ and $\Theta(n^2)$ for $KT_1$.

Our work leaves several interesting directions for future research.
Most notably, it would be interesting to understand how much randomization can help.
If randomized algorithms are allowed, how small can the minimum required window size be made, and how much can the exploration time be improved?
More broadly, closing the remaining polylogarithmic gap in the general case and clarifying the exact advantage of additional local visibility also remain interesting open problems.

\paragraph{AI methodology}
We used Codex (ChatGPT-6 Astra) to assist with proofreading, refining mathematical expressions, and checking notational consistency.
Codex also generated the supplementary lemmas in Appendix~\ref{sec:supplemental_lemmas}.
The authors verified all content, including the references, and take full responsibility for its correctness and originality.

\clearpage

\bibliographystyle{alpha}
\bibliography{tvg}

\clearpage


\appendix
\section{Supplemental Lemmas}
\label{sec:supplemental_lemmas}

\begin{lemma}
  \label{lem:inequality}
  For any real number $x > 0$ and $n \ge 3$, we have $\ln n < x n^{1/(2x)}$.
\end{lemma}

\begin{proof}
  For $n \ge 3$, let $y=\frac{\ln n}{2x} > 0$. Then
  \[
    \ln n < x n^{1/(2x)}
    \quad\Longleftrightarrow\quad
    2y<e^y.
  \]
  Define $h(y)=e^y-2y$. Then $h'(y)=e^y-2$, so $h$ attains its minimum on $(0,\infty)$ at $y=\ln 2$, where
  \[
    h(\ln 2)=2-2\ln 2>0.
  \]
  Hence $e^y>2y$ for all $y>0$, proving the claim.
\end{proof}

\begin{lemma}
  \label{lem:epsilon_monotone}
  Let
  \[
    f(m,n)=\frac{m\ln n}{1+\ln m-\ln n}.
  \]
  Then, $f(m,n)$ is monotonically increasing in both $m$ and $n$ for all
  $n \ge 3$ and $m \ge n$.
\end{lemma}

\begin{proof}
  Let
  \[
    d = 1+\ln m-\ln n.
  \]
  Since $m \ge n \ge 3$, we have $d \ge 1 > 0$.

  First, treating $n$ as a constant,
  \[
    \frac{\partial f}{\partial m}
    =
    \frac{\ln n\,(d-1)}{d^2}
    =
    \frac{\ln n\,(\ln m-\ln n)}{d^2}
    \ge 0,
  \]
  because $\ln n>0$ and $\ln m \ge \ln n$.

  Next, treating $m$ as a constant,
  \[
    \frac{\partial f}{\partial n}
    =
    \frac{m(1+\ln m)}{n\,d^2}
    >0,
  \]
  because $m>0$, $n>0$, and $\ln m \ge \ln n \ge \ln 3 >0$.

  Hence, $f(m,n)$ is monotonically increasing in both variables.
\end{proof}

\section{Charge Bounds in the 0-Hop Model}
\label{app:kt0_charges}

Fix a round $r$ that ends with $\vtarget\neq\emptyset$.
We use the charges $\iota'_V$, $\iota'_E$ and their sums $I_V$, $I_E$ defined in the proof of Lemma~\ref{lem:open_decrease}.
Let $W$ be the set of known nodes at the end of the round and let $n'_r=|W|$.
Then $n'_r\le n_r$, where $n_r$ is the value of the estimate $n'$ at the end of the round.
Let $T_r$ be a spanning tree present at every observation in this round.

\begin{observation}
  \label{obs:kt0_tree_edges}
  Throughout round $r$, every known endpoint of an unused edge of $T_r$ is a target, and every previously traversed edge of $T_r$ remains in the map.
\end{observation}

\begin{proof}
  Every edge of $T_r$ is present at every observation in the round, so neither of its ports is excluded.
  If the edge is unused, a known endpoint is either unvisited in this round or has observed its port, which is open and unexcluded; in either case, it is a target.
  If the edge has been traversed, both port assignments are known, so the edge remains in the map.
\end{proof}

\begin{lemma}
  \label{lem:kt0_nodes}
  $I_V=O(n_r\log n_r)$.
\end{lemma}

\begin{proof}
  Fix a component $S$ of $T_r[W]$.
  Order its nodes $v_1,\ldots,v_{|S|}$ by when they first belong to
  $\vmap\setminus\vtarget$ during the round, placing nodes that never
  do last in an arbitrary order.
  If both endpoints enter this set during the same move,
  place the source first.
  If $W\neq V$, an unused edge of $T_r$ joins $S$ to $V\setminus W$, so $S$ contains a target at the end of the round by Observation~\ref{obs:kt0_tree_edges}.
  This also holds when $W=V$ by the assumption $\vtarget\neq\emptyset$.
  Hence the last node has zero charge.

  Immediately after $v_i$ receives a positive charge, follow the path in $S$ to any later node $v_j$.
  If the path lies in the map, $v_j$ is known and, by the ordering, is still a target.
  Otherwise, the first missing edge is unused, so its preceding endpoint is a target by Observation~\ref{obs:kt0_tree_edges}.
  Hence $\iota'_V(v_i)\le\min_{j>i}d_S(v_i,v_j)$ for $i<|S|$.
  Lemma~\ref{lem:tee_dist}, applied to each component, now gives
  \[
    I_V\le\sum_S2|S|\log|S|\le2n'_r\log n'_r=O(n_r\log n_r).\qedhere
  \]
\end{proof}

\begin{lemma}
  \label{lem:kt0_edges}
  $I_E=O(\almcon(n_r,m_r)m_r+n_r\log^2n_r)$.
\end{lemma}

\begin{proof}
  For integers $j\ge1$ and $k\ge1$, let
  \[
    F_r(j)=|\{e\in E:\iota'_E(e)\ge j\}|,\qquad
    X_k=\{e\in E:\iota'_E(e)\ge2k-1\}.
  \]
  Extend $T_r[W]$ to a tree $T'_r$ on $W$ as follows.
  If $T_r[W]$ is disconnected, connectivity of $T_r$ ensures that each component contains a node adjacent in $T_r$ to a node in $V\setminus W$.
  Choose one such node in each component as its \emph{gate}, and join the gates by a star of \emph{auxiliary edges}.
  By Observation~\ref{obs:kt0_tree_edges}, the gates remain targets whenever they are in $\vmap$.
  No edge of $T'_r$ receives an edge charge: original tree edges remain in the map once traversed, and auxiliary edges join targets, where no edge charges are assigned.
  Hence $E(T'_r)\cap X_k=\emptyset$.

  We show that $T'_r\cup X_k$ has girth at least $2k+1$.
  Suppose, for a contradiction, that it contains a cycle $C$ of length $\ell\le2k$.
  Since $T'_r$ is a tree, $C\cap X_k\neq\emptyset$.
  Let $e$ be the first edge of $C\cap X_k$ deleted from the map, breaking ties by the fixed deletion order.
  In the intermediate map just before this deletion, follow $C-e$ from the agent's location, an endpoint of $e$.
  If this path lies entirely in the map, replacing $e$ in a shortest path to a target adds at most $\ell-2$ steps.
  Otherwise, let $\{x,y\}$ be its first missing edge, with $x$ preceding $y$.
  The prefix ending at $x$ remains in the map after deleting $e$ and has length at most $\ell-2$.
  If $\{x,y\}$ is auxiliary, then $x$ is a gate and hence a target.
  If it is an original tree edge, it is unused, so $x$ is a target by Observation~\ref{obs:kt0_tree_edges}.
  Otherwise, $\{x,y\}\in C \cap X_k$.
  The edge $\{x,y\}$ is absent from the map but must enter it later to receive its charge.
  Since exclusions persist within the round, neither port can already be excluded.
  Hence the edge is unused; otherwise, it would be in the map.
  Thus $x$ is a target.
  Whether or not $C-e$ lies entirely in the map, we obtain $\iota'_E(e)\le\ell-2\le2k-2$, contradicting $e\in X_k$.

  Subtracting the $n'_r-1$ tree edges from the bound in Theorem~\ref{thm:girth} gives
  \[
    F_r(2k-1)=|X_k|
    \le \frac{(n'_r)^{1+1/k}-n'_r}{2}+1
    \le \frac{n_r^{1+1/k}-n_r}{2}+1,
  \]
  using $n'_r\le n_r$.

  Only traversed edges receive charges, so $F_r(j)\le m_r$; also, $F_r(n_r)=0$ since each charge is at most $n_r-1$.
  Since $n_r\le m_r\le\binom{n_r}2$, the numerical summation in the proof of Lemma~\ref{lem:edges} applies directly with $(n,m)=(n_r,m_r)$ and gives
  \[
    I_E=\sum_{j\ge1}F_r(j)
    =O(\almcon(n_r,m_r)m_r+n_r\log^2n_r).\qedhere
  \]
\end{proof}

\end{document}